\documentclass[runningheads]{llncs}
\usepackage{amssymb,amsmath,graphicx,verbatim,wrapfig}
\usepackage[tight,hang]{subfigure}
\usepackage{plaatjes}
\usepackage {color}
\usepackage {microtype}
\usepackage{mathptmx}
\usepackage{compress}
\DeclareMathAlphabet{\mathcal}{OMS}{cmsy}{m}{n}
\usepackage{cite,hyperref}
\usepackage {plaatjes}
\usepackage{wrapfig}
\usepackage {a4wide}

\graphicspath{{figures/}}

\spnewtheorem{observation}{Observation}{\bfseries}{\itshape}
\definecolor {remarkcolor} {rgb} {0.1,0.7,0.2}

\let\doendproof\endproof
\renewcommand\endproof{~\hfill\qed\doendproof}

\definecolor {results_green}  {rgb} {0.0,0.6,0.1}
\definecolor {results_red}    {rgb} {0.6,0.0,0.1}
\definecolor {results_orange} {rgb} {0.9,0.3,0.0}

\newcommand {\R} {\mathbb {R}}
\newcommand {\D} {\mathbb {D}}

\begin{document}
\mainmatter

\title {Adjacency-Preserving Spatial Treemaps}

\author{
Kevin Buchin\inst{1}
\and
David Eppstein\inst{2}
\and
Maarten L\"offler\inst{2}
\and \\
Martin N\"ollenburg\inst{3}
\and
Rodrigo I. Silveira\inst{4}
}

\authorrunning{Buchin et al.}

\institute{
Dept. of Mathematics and Computer Science, TU Eindhoven.\\
\and
Dept. of Computer Science, University of California, Irvine.\\
\and
Institute of Theoretical Informatics, Karlsruhe Institute of Technology.\\
\and
Dept. de Matem\`atica Aplicada II, Universitat Polit\`ecnica de Catalunya.\\
}


\maketitle

\vspace {-5pt}
  \begin {abstract}
Rectangular layouts, subdivisions of an outer rectangle into smaller rectangles, have many applications in visualizing spatial information, for instance in rectangular cartograms in which the rectangles represent geographic or political regions. A \emph{spatial treemap} is a rectangular layout with a hierarchical structure: the outer rectangle is subdivided into rectangles that are in turn subdivided into smaller rectangles. We describe algorithms for transforming a rectangular layout that does not have this hierarchical structure, together with a clustering of the rectangles of the layout, into a spatial treemap that respects the clustering and also respects to the extent possible the adjacencies of the input layout.
  \end {abstract}

\section {Introduction}

Spatial treemaps are an effective technique to visualize two-dimensional hierarchical information.
They display hierarchical data by using nested rectangles in a space-filling layout.
Each rectangle represents a geometric or geographic region, which in turn can be subdivided recursively into smaller regions. On lower levels of the recursion, rectangles can also be subdivided based on non-spatial attributes. Typically, at the lowest level some attribute of interest of the region is summarized by using properties like area or color.
Treemaps were originally proposed to represent one-dimensional information in two dimensions~\cite{s-tvtm-92}. However, they are well suited to represent spatial---two-dimensional---data because the containment metaphor of the nested rectangles has a natural geographic meaning, and two-dimensional data makes an efficient use of space~\cite{wd-sot-08}.


Spatial treemaps are closely related to rectangular cartograms~\cite{r-rsc-34}: distorted maps where each region is represented by a rectangle whose area corresponds to a numerical attribute such as population. Rectangular cartograms can be seen as spatial treemaps with only one level; multi-level spatial treemaps in which every rectangle corresponds to a region are also known as \emph{rectangular hierarchical cartograms}~\cite{sdw-chl-09,sdw-rhc-10}.
Spatial treemaps and rectangular cartograms have in common that it is essential to preserve the recognizability of the regions shown~\cite{ks-rc-07}.
Most previous work on spatial treemaps reflects this by focusing on the preservation of distances between the rectangular regions and their geographic counterparts (that is, they minimize the displacement of the regions).
However, often small displacement does not imply recognizability (swapping the position of two small neighboring countries can result in small displacement, but a big loss of recognizability).
In the case of cartograms, most emphasis has been put on preserving adjacencies between the geographic regions.
It has also been shown that while preserving the topology it is possible to keep the displacement error small~\cite{ks-rc-07,bsv-orel-11}.

In this paper we are interested in constructing high-quality spatial treemaps by prioritizing the preservation of topology, following a principle already used for rectangular cartograms. Previous work on treemaps has recognized that preserving neighborhood relationships and relative positions between the regions were important criteria~\cite{hkps-rma-04,mknrs-vant-07,wd-sot-08}, but we are not aware of treemap algorithms that put the emphasis on preserving topology.

\eenplaatje [trim = 0mm 39mm 70mm 0mm,clip=true,width=0.5\textwidth]
{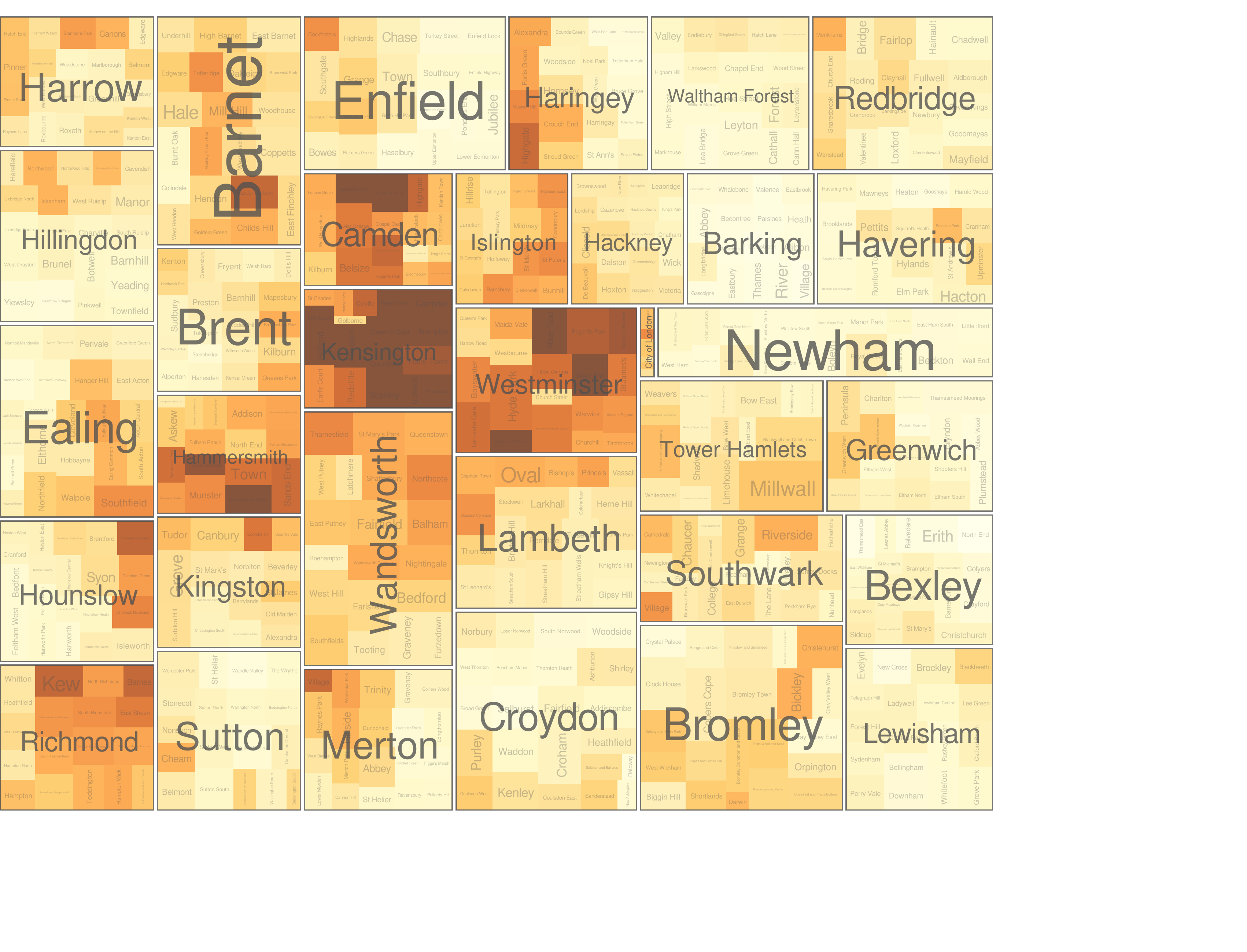}
{A 2-level spatial treemap from~\cite{sdw-chl-09}; used with permission.}

The importance of preserving adjacencies in spatial treemaps can be appreciated by viewing a concrete example.
Figure~\ref{fig:bor-war}, from~\cite{sdw-chl-09}, shows a spatial treemap of property transactions in London between 2000 and 2008, with two levels formed by the boroughs and wards of London and colors representing average prices. To see whether housing prices of neighboring wards are correlated, it is important to preserve adjacencies: otherwise it is easy to draw incorrect conclusions, like seeing clusters that do not actually exist, or missing existing ones.

Preserving topology in spatial treemaps poses different challenges than in (non-hierarchical) rectangular cartograms.
Topology-preserving rectangular cartograms exist under very mild conditions and can be constructed efficiently~\cite{bsv-orel-11,ks-rc-07}.
As we show in this paper, this is not the case when a hierarchy is added to the picture.

In this paper we consider the following setting:
the input is a hierarchical rectangular subdivision with two levels.
We consider only two levels due to the complexity of the general $m$-level case.
However, the two-level case is interesting on its own, and
applications that use only two-level data have recently appeared~\cite{sdw-chl-09}.

Furthermore, we adopt a 2-phase approach for building spatial treemaps.
In the first phase, a base rectangular cartogram is produced from the original geographic regions.
This can be done with one of the many algorithms for rectangular cartograms~\cite{bsv-orel-11}.
The result will contain all the bottom-level regions as rectangles, but the top-level regions will not be rectangular yet, thus will not represent the hierarchical structure.
In the second phase, we convert the base cartogram into a treemap by making the top-level regions rectangles.
It is at this stage that we intend to preserve the topology of the base cartogram as much as possible, and where our algorithms come in.
See Figure~\ref {fig:ex-input+ex-output} for an example.

 \tweeplaatjes[scale=0.8] {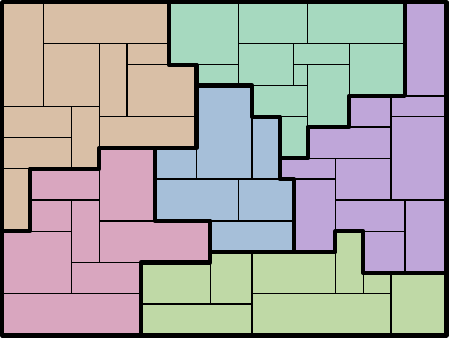} {ex-output} {(a) An example input: a full layout of the bottom level, but the regions at a higher level in the hierarchy are not rectangles. (b) The desired output: another layout, in which as many lower-level adjacencies as possible have been kept while reshaping the regions at a higher level into rectangles.}

The advantage of this 2-phase approach is that it allows for customization and user interaction.
Interactive exploration of the data is essential when visualizing large amounts of data.
The freedom to use an arbitrary rectangular layout algorithm in the first phase of the construction allows the user to prioritize the adjacencies that he or she considers most essential.
In the second phase, our algorithm will produce a treemap that will try to preserve as many as the adjacencies in the base cartogram as possible.

In addition, we go one step further and consider preserving the \emph{orientations} of the adjacencies in the base cartogram (that is, whether two neighboring regions share a vertical or horizontal edge, and which one is on which side).
This additional constraint is justified by the fact that the regions represent geographic or political regions, and relative positions between regions are an important factor when visualizing this type of data~\cite{bsv-orel-11,ks-rc-07}.
The preservation of orientations has been studied for cartograms~\cite{em-ocrl-09}, but to our knowledge, this is the first time they are considered for spatial treemaps.

  We can distinguish three types of adjacency-relations: (i) top-level adjacencies, (ii) internal bottom-level adjacencies (adjacencies between two rectangles that belong to the same top-level region), and (iii) external bottom-level adjacencies (adjacencies between two rectangles that belong to different top-level regions).
  As we argue in the next section, we can always preserve all adjacencies of types (i) and (ii) under a mild assumption, hence the objective of our algorithms is to construct treemaps that preserve as many adjacencies of type (iii) as possible.
We consider several variants of the problem, based on whether the orientations of the adjacencies have to be preserved, and whether the top-level layout is given in advance.
In order to give efficient algorithms, we restrict ourselves to top-level regions that are orthogonally convex.
This is a technical limitation that seems difficult to overcome, but that we expect  does not limit the applicability of our results too much: our algorithms should still be useful for many practical instances, for example, by subdividing non-convex regions into few convex pieces.

  \paragraph {\bf Results}
In the most constrained case in which adjacencies and their orientations need to be preserved and the top-level layout is given, we solve the problem in $O(n)$ time, where $n$ is the total number of rectangles.
The case in which the global layout is not fixed is much more challenging: it takes a combination of several techniques based on regular edge labelings to obtain an algorithm that solves the problem optimally in $O(k^4 \log k + n)$ time, for $k$ the number of top-level regions;
we expect $k$ to be much smaller than $n$.
Finally, we prove that the case in which the orientations of adjacencies do not need to be preserved is NP-hard; we give worst-case bounds and an approximation algorithm.

\section {Preliminaries}

  \paragraph {\bf Rectangles and Subdivisions}


    All geometric objects like rectangles and polygons in this paper are defined as rectilinear (axis-aligned) objects in the Euclidean plane $\R^2$.
    We work in the Euclidean plane $\R^2$.
    A \emph {rectangle} in this text is always an axis-aligned rectangle.
    We also define polygon, convex, etc. in a rectilinear sense.
    A~set of rectangles $\cal R$ is called a \emph {rectangle complex}
    if the interiors of none of the rectangles overlap, and each pair of
    rectangles is either completely disjoint or shares part of an edge; no two rectangles may meet in a single point. Each
    rectangle of a rectangle complex is a \emph{cell} of that complex. We represent
    rectangle complexes using a structure that has bidirectional
    pointers between neighboring cells.

    Let $\cal R$ be a rectangle complex. The \emph{boundary} of $\cal
    R$ is the boundary of the the union of the rectangles in $\cal
    R$.
    Note that this is always a proper polygon, but it could have multiple components and holes.  We say that
    $\cal R$ is
    \emph{simple} if its boundary is a simple polygon, i.e., it is
    connected and has no holes.  We say that $\cal R$ is \emph{convex}
    if its boundary is orthogonally convex, i.e., the intersection of
    any horizontal or vertical line with $\cal R$ is either empty or a
    single line segment. We say that $\cal R$ is \emph{rectangular} if
    its boundary is a rectangle.

    Let $\cal R'$ be another rectangle complex. We say that $\cal R'$ is an \emph {extension} of $\cal R$ if there is a
    bijective mapping between the cells in $\cal R$ and $\cal R'$ that preserves the adjacencies and their orientations. Note that $\cal R'$ could have adjacencies not present in
    $\cal R$ though.  We say that $\cal R'$ is a \emph {simple
      extension} of $\cal R$ if $\cal R$ is not simple but $\cal R'$ is;
    similarly we may call it a \emph {convex extension} or a \emph
    {rectangular extension}. 
    
    We show that every rectangle complex has a rectangular extension.
    \begin{lemma} \label{lem:extension}
      Let $\cal R$ be a rectangle complex. There always exists a rectangular extension of $\cal R$.
    \end {lemma}
    \begin {proof}
    We first augment $\cal R$ by four rectangles forming a bounding box of $\cal R$. Our goal is to extend the complex so that no holes inside the bounding box remain, while all existing adjacencies are preserved with their orientation. Obviously each hole is formed by at least four adjacent rectangles. Let $H$ be a hole of the augmented complex. If there is a rectangle $R$ adjacent to $H$ with a full rectangle edge, we can extend $R$ into the hole until it touches another rectangle. This either closes the hole or splits it into holes of lower complexity. Now let's assume that there is a hole without a rectangle adjacent to it along a full edge. Then each edge of the hole is a partial rectangle edge blocked by another rectangle. This is only possible in a ``windmill'' configuration of four rectangles cyclically blocking each other. But such a hole can be removed by moving two opposite rectangles toward each other while shrinking the other two rectangles. See Figure~\ref{fig:windmill}. None of the operations removes adjacencies or changes their orientations.
    \end {proof}
    
    \eenplaatje {windmill} {Removing a windmill hole.}

    We define $\D = \{\mathrm{left}, \mathrm{right}, \mathrm{top},
    \mathrm{bottom}\}$ to be the set of the four cardinal directions. For a
    direction $d \in \D$ we use the notation $-d$ to refer to the
    direction opposite from $d$.  We
    define an object $O \subset \R^2$ to be \emph {extreme} in direction
    $d$ with respect to a rectangle complex $\cal R$ if there is a point
    in $O$ that is at least as far in direction $d$ as any point in
    $\cal R$.  Let $R \in \cal R$ be a cell, and $d \in \D$ a
    direction. We say $R$ is \emph {$d$-extensible} if there exists a
    rectangular extension $\cal R'$ of $\cal R$ in which $R$ is extreme
    in direction $d$ with respect to $\cal R'$ (or in other words, if
    its $d$-side is part of the boundary of $\cal R'$).

    A set of simple rectangle complexes $\cal L$ is called a
    (rectilinear) \emph{layout} if the boundary of the union of all
    complexes is a rectangle, the interiors of the complexes are
    disjoint, and no point in $\cal L$ belongs to more than three
    cells. If all complexes are rectangular we say that $\cal L$ is a
    \emph{rectangular layout}. We call the rectangle bounding $\cal L$
    the \emph{root box}.

    Let $\cal L$ be a rectilinear layout. We define the
    \emph{global layout} $\cal L'$ of $\cal L$ as the subdivision of the
    root box of $\cal L$, in which the (\emph{global}) regions are
    defined by the boundaries of the complexes in $\cal L$. We say $\cal
    L'$ is \emph{rectangular} if all regions in $\cal L'$ are rectangles.

  \paragraph {\bf Dual Graphs of Rectangle Complexes}

   The \emph {dual graph} of a rectangular complex is an embedded planar graph with one vertex for every rectangle in the complex, and an edge between two vertices if the corresponding rectangles touch (have overlapping edge pieces).
    The \emph {extended dual graph} of a rectangular complex with a rectangular boundary has four additional vertices for the four sides of the rectangle, and an edge between a normal vertex and an additional vertex if the corresponding rectangle touches the corresponding side of the bounding box.
    We will be using dual graphs of the whole rectangular layout, of individual complexes, and of the global layout (ignoring the bottom level subdivision); Figure~\ref {fig:bottomlevel-rects+bottomlevel-graph+toplevel-regions+toplevel-graph} shows some examples.
   Extended dual graphs of rectangular rectangle complexes are fully triangulated (except for the outer face which is a
    quadrilateral), and the graphs that can arise in this way are characterized by the following lemma~\cite{kh-tafrd-93,kk-rdpg-85,ks-rc-07}:

    \vierplaatjes {bottomlevel-rects} {bottomlevel-graph} {toplevel-regions} {toplevel-graph} {(a) A bottom level rectangle complex. (b) The dual graph of the complex. (c) A global layout. (d) The extended dual graph of the global layout.}

    \begin {lemma} \label {lem:no-sep-tri}
      A triangulated plane graph $G$ with a quadrilateral outer face is the dual graph of a rectangular rectangle complex if and only if $G$ has no separating triangles.
    \end {lemma}

    Now, consider the three types of adjacencies we wish to preserve:
    1) (top-level) adjacencies between global regions, 2) internal
    (bottom-level) adjacencies between
    the cells in one rectangle complex, and 3) external (bottom-level) adjacencies
    between cells of adjacent rectangle complexes.

    { \observation
      It is always possible to keep all internal bottom-level adjacencies.
    }

    { \observation It is possible to keep all top-level adjacencies if
      and only if the extended dual graph of the global input layout has no separating triangles.
    }

Observation 1 follows by applying Lemma~\ref {lem:extension} to all regions, and Observation~2
follows from  Lemma~\ref {lem:no-sep-tri} since the extended dual graph of the global regions is fully triangulated.

    \eenplaatje {corner-lock} {Not all external adjacencies can be kept.}

     From now on we assume that the dual graph of the global regions has no separating triangles, and we will preserve all adjacencies of types 1 and 2. Unfortunately, it is not always possible to keep adjacencies of type 3---see Figure~\ref {fig:corner-lock}---and for every adjacency of type 3 that we fail to preserve, another adjacency that was not present in the original layout will appear. Therefore, our aim is to preserve as many of these adjacencies as possible.


\section{Preserving orientations}\label{sec:preserve}

We begin studying the version of the problem where all internal
adjacencies have to be preserved respecting their original
orientations.  Additionally, we want to maximize the number of
preserved and correctly oriented (bottom-level) external adjacencies.
We consider two scenarios: first we assume that the global layout is
part of the input, and then we study the case in which we optimize
over all global layouts. The former situation is particularly
interesting for GIS applications, in which the user specifies a
certain global layout that needs to be filled with the bottom-level
cells. If, however, the bottom-level adjacencies are more important,
then optimizing over global layouts allows to preserve more external
adjacencies.

\subsection{Given the global layout}

In this section we are given, in addition to the initial two-level
subdivision $\cal L$, a global target layout $\cal L'$.  The goal is to find a two-level treemap that preserves
all oriented bottom-level internal adjacencies and that
maximizes the number of preserved oriented bottom-level external adjacencies in the output.

First observe that in the rectangular output layout any two
neighboring global regions have a single orientation for their
adjacency. Hence we can only keep those bottom-level external
adjacencies that have the same orientation in the input as their
corresponding global regions have in the output layout.  Secondly,
consider a rectangle $R$ in a complex $\cal R$, and a rectangle $B$ in
another complex $\cal B$. Observe that if $R$ and $B$ are adjacent in
the input, for example with $R$ to the left of $B$, then their
adjacency can be preserved only if $R$ is right-extensible in $\cal R$
and $B$ is left-extensible in $\cal B$.

The main result in this section is that the previous two conditions
are enough to describe all adjacencies that cannot be preserved,
whereas all the other ones can be kept.  Furthermore, we will show how
to decide extensibility for convex complexes, and how to construct a
final solution that preserves all possible adjacencies, leading to an
algorithm for the optimal solution.

Recall that we assume all regions are orthogonally
convex. 
Consider each rectangle complex of $\cal L$ separately. Since we know the required global
layout and since all cells externally adjacent to our
region are consecutive along its boundary, we can
immediately determine the cells on each of the four sides of the
output region (see Figure~\ref
{fig:green-input+green-global+green-extensions+green-output}). The
reason is that for a rectangle $R$ that is exterior to its region
$\cal R$, and that is adjacent to another rectangle $B \in \cal B$,
their adjacency is relevant only if $\cal R$ and $\cal B$ are adjacent
with the same orientation in the global layout.
We can easily categorize the extensible rectangles of a convex
rectangle complex.

    \vierplaatjes {green-input} {green-global} {green-extensions} {green-output}
    { (a) A region in the input.
      (b) The same region in the given global layout.
      (c) Edges of rectangles that want to become part of a boundary have been marked with arrows. Note that one rectangle wants to become part of the top boundary but can't, because it is not extensible in that direction.
      (d) All arrows that aren't blocked can be made happy.
    }

\newcommand{\textlemmaextensible}{%
  Let $\cal R$ be a convex rectangle complex, let $R \in \cal R$
  be a rectangle, and $d \in \D$ a direction.  $R$ is
  $d$-extensible if and only if there is no rectangle $R' \in \cal
  R$ directly adjacent to $R$ on the $d$-side of $R$.}

    \begin {lemma}\label{lem:extensible}
      \textlemmaextensible
    \end {lemma}

\begin {proof}
  For the `only if' part, simply note that if there is such a
  rectangle $R' \in \cal R$, then the adjacency between $R$ and
  $R'$ must be preserved, with its original orientation. Hence
  there is always a point in $R'$ that is further in direction $d$
  than any point in $R$. So $R$ is not $d$-extensible.

  For the `if' part, consider the complex obtained by extending
  $R$ in direction $d$ until it becomes extreme in that
  direction. This is always possible because $\cal R$ is convex;
  the resulting complex is still simple. Now we add a temporary
  bounding box consisting of four rectangles around $\cal R$, one
  in each direction, such that $R$ is adjacent to the one on the
  $d$-side.  Then we can apply Lemma~\ref{lem:extension} and find
  a rectangular extension of $\cal R$ where $R$ is extreme on the
  $d$-side.
\end {proof}

    Unfortunately, though, we cannot extend all extensible rectangles
    at the same time. However, we show that we can actually extend
    all those rectangles that we want to extend for an optimal solution.

    We call a rectangle of a certain complex belonging to a global region
    \emph {engaged} if it
    wants to be adjacent to a rectangle of another global region, and the
    direction of their desired adjacency is the same as the direction
    of the adjacency between these two regions in the global
    layout.  We say it is $d$-engaged if this direction is $d \in \D$.

    Therefore, the rectangles that we want to extend are exactly those
    that are $d$-extensible and $d$-engaged, since they are the only
    ones that help preserve bottom-level exterior adjacencies.  It
    turns out that extending all these rectangles is possible, because
    the engaged rectangles of $\cal R$ have a special property: 

\newcommand{\textlemmagivengloballayout}{%
If we walk around the boundary of a region $\cal R$, we
    encounter all $d$-engaged rectangles consecutively.}
    \begin{lemma}\label{lem:given-global-layout}
      \textlemmagivengloballayout
    \end{lemma}

\begin{proof}
  Suppose that when walking clockwise along the boundary of $\cal R$
  we encounter rectangles $R_1, R_2, R_3$ that are $d$-,$d'$-, and
  $d$-engaged, respectively.  Since $R_1$ and $R_3$ are both
  $d$-engaged, in the global layout they have the same direction of
  external adjacency.  However, if $d' \neq d$, then $R_2$ has a
  different direction, implying that in the global layout this is also
  the same way.  This contradicts the fact that in the global layout
  $\cal R$ is a rectangle, so the rectangles engaged in the four
  different directions appear contiguously.
\end{proof}

This property of $d$-engaged rectangles is useful due to the following
fact.

\newcommand{\textlemmacontiguous}{%
Let $\cal R$ be a convex rectangle complex composed of $r$
rectangles, and let $S$ be a subset of the extensible and
engaged rectangles in $\cal R$ with the property that if we
order them according to a clockwise walk along the boundary of $\cal
R$, all $d$-extensible rectangles in $S$ are encountered
consecutively for each $d \in \D$ and in the correct clockwise
order. We can compute, in $O(r)$ time, a rectangular extension
$\cal R'$ of $\cal R$ in which all $d$-extensible rectangles in
$S$ are extreme in direction $d$, for all $d \in
\D$. 
}

    \begin {lemma}
      \label{lem:contiguos}
      \textlemmacontiguous
    \end {lemma}

\begin {proof}
  We use the same idea as in the proof of
  Lemma~\ref{lem:extensible}, but now we extend all rectangles in
  $S$ at the same time.  Since by
  Lemma~\ref{lem:given-global-layout} $d$-engaged rectangles
  appear consecutively around the boundary of $\cal R$, there
  cannot be any conflicts preventing the extension of $d$-engaged
  rectangles: rectangles with the same direction extend all toward
  the same side, thus they can all be made extreme.  On the other
  hand, two rectangles extended toward different directions cannot
  influence each other because that would imply that the
  directions do not appear contiguously or in the wrong order.

  It remains to apply Lemma~\ref{lem:extension} and show that the
  rectangular extension can be found in linear time. Since $\cal
  R$ with the rectangles in $S$ extended is still a simple
  polygon, all holes of the complex after augmenting it by the
  four external rectangles are adjacent to the external
  rectangles. Hence there are no windmill holes. We can then start
  walking clockwise along the boundary of $\cal R$ at the first
  $d$-extended rectangle and extend all $d$-extensible rectangles
  until we reach the first $d+1$-extended rectangle. None of these
  rectangles is blocked in direction $d$. We close the corner
  between the $d$-side and the $d+1$-side by extending either the
  last $d$-extended rectangle in direction $d+1$ or vice versa,
  which is always possible since they cannot both block each
  other. We continue this process along all four sides of $\cal
  R$. Let $H$ be a remaining hole on the $d$-side. It has the
  property that none of its adjacent rectangles is $d$-extensible
  and that it is bounded by two staircases. We can then close the
  hole in linear time by simultaneously walking along the two
  staircases and maximally extending rectangles orthogonally to
  direction $d$.
\end {proof}

    Therefore, the engaged and extensible rectangles form a subset of
    rectangles for which Lemma~\ref{lem:contiguos} holds, thus by
    using the lemma we can find a rectangular extension where all
    extensible and engaged rectangles are extreme in the appropriate
    direction.

    Then we can apply this idea to each region.  Now we still have to
    match up the adjacencies in an optimal way, that is, preserving as
    many adjacencies from the input as possible.  This can be done by
    matching horizontal and vertical adjacencies independently.  It is
    always possible to get all the external bottom-level adjacencies
    that need to be preserved.  This can be seen as follows (see also
    Figure~\ref {fig:green-re-global+green-re-graph+green-re-out}).
    We process first all horizontal adjacencies.  Consider a complete
    stretch of horizontal boundary in the global layout.  Then the
    position and length of the boundary of each region adjacent to
    that boundary are fixed, from the global layout.  The only freedom
    left is in the $x$-coordinates of the vertical edges of the
    rectangles that form part of that boundary (except for the
    leftmost and rightmost borders of each region, which are also
    fixed).  Since the adjacencies that want to be preserved are part
    of the input, it is always possible to set the $x$-coordinates in
    order to fulfill them all.  The same can be done with all
    horizontal boundaries.  The vertical boundaries are independent,
    thus can be processed in exactly the same way. This yields the
    main theorem in this subsection. 

    \drieplaatjes {green-re-global} {green-re-graph} {green-re-out} {(a) After we solved all the different colors separately, we don't necessarily have the right adjacencies yet. (b) We can indicate their desired adjacencies that are still possible (so, the adjacencies between two edges of rectangles that actually ended up on the outside) as a graph. Note that this graph is planar. (c) We can poke around in the insides of the rectangles to make all desired adjacencies happen.}

    \begin {theorem}
      Let $\cal T$ be a $2$-level treemap, where $n$ is the number of
      cells in the bottom level, and where all global regions are
      orthogonally convex. For a given global target layout $\cal L$, we can
      find, in $O(n)$ time, a rectangular layout of $\cal T$ that
      respects $\cal L$, preserves all oriented internal bottom-level
      adjacencies, and preserves as many oriented external
      bottom-level adjacencies as possible.
    \end {theorem}

\subsection {Unconstrained global layout}

In this section the global target layout of the rectangle complexes
is not given, i.e., we are given a rectilinear input layout and need
to find a rectangular output layout preserving all adjacencies of the
rectangle complexes and preserving a maximum number of
adjacencies of the cells of different complexes.

We can represent a particular rectangular global layout $\cal L$ as a
\emph{regular edge labeling}~\cite{kh-relfc-97} of the dual
graph $G(\cal L)$ of the global layout.
Let $G(\cal L)$ be
the extended dual graph of $\cal L$. Then $\cal L$
induces an edge labeling as follows: an edge corresponding to a joint
vertical (horizontal) boundary of two rectangular complexes is colored
blue (red). Furthermore, blue edges are directed from left to right
and red edges from bottom to top. Clearly, the edge labeling obtained
from $\cal L$ in this way satisfies that around each inner vertex $v$
of $G(\cal L)$ the incident edges with the same color and the same
direction form contiguous blocks around $v$. The edges incident to one
of the external vertices $\{l,t,r,b\}$ all have the same label. Such
an edge labeling is called \emph{regular}~\cite{kh-relfc-97}. Each
regular edge labeling of the extended dual graph $G(\cal L)$
defines an equivalence class of global layouts.

In order to represent the family of all possible rectangular global
layouts we apply a technique described by Eppstein et
al.~\cite{emsv-aurl-09,em-ocrl-09}. Let $\cal L$ be the rectilinear
global input layout and let $G(\cal L)$ be its extended dual
graph. The first step is to decompose $G(\cal L)$ by its separating
4-cycles into minors called \emph{separation components} with the
property that they do not have non-trivial separating 4-cycles any
more, i.e., 4-cycles with more than a single vertex in the inner part
of the cycle. If $C$ is a separating 4-cycle the interior separation
component consists of $C$ and the subgraph induced by the vertices
interior to $C$. The outer separation component is obtained by
replacing all vertices in the interior of $C$ by a single vertex
connected to each vertex of $C$. This decomposition can be obtained in linear
time~\cite{emsv-aurl-09}. We can then treat each component in the
decomposition independently and finally construct an optimal
rectangular global layout from the optimal solutions of its
descendants in the decomposition tree. So let's consider a single
component of the decomposition, which by construction has no
non-trivial separating 4-cycles.

%
%

\subsubsection {Preprocessing of the bottom level}

We start with a preprocessing step to compute the number of realizable
external bottom-level adjacencies for pairs of adjacent global
regions. This allows us to ignore the bottom-level cells in later
steps and to focus on the global layout and orientations of global
adjacencies.

Let $\cal L$ be a global layout, let $\cal R$ and $\cal S$ be two
adjacent rectangle complexes in $\cal L$, and let $d \in \D$ be an
orientation.  Then we define $\omega(\mathcal{R},\mathcal{S},d)$ to be
the total number of adjacencies between $d$-engaged and $d$-extensible
rectangles in $\mathcal{R}$ and $-d$-engaged and $-d$-extensible
rectangles in $\mathcal{S}$. By Lemma~\ref{lem:contiguos} there is a
rectangular layout of $\cal R$ and $\cal S$ with exactly
$\omega(\mathcal{R},\mathcal{S},d)$ external bottom-level adjacencies
between $\cal R$ and $\cal S$.

  We show the following (perhaps surprising) lemma:

\newcommand{\textlemmaweights}{
For any pair $\cal L$ and $\cal L'$ of global layouts and any pair
$\cal R$ and $\cal S$ of rectangular rectangle complexes, whose
adjacency direction with respect to $\cal R$ is $d$ in $\cal L$
and $d'$ in $\cal L'$ the number of external bottom level
adjacencies between $\cal R$ and $\cal S$ in any optimal solution
for $\cal L'$ differs by $\omega (\mathcal{R},\mathcal{S},d') -
\omega (\mathcal{R},\mathcal{S},d)$ from $\cal L$. For adjacent
rectangle complexes whose adjacency direction is the same in both
global layouts the number of adjacencies in any optimal solution
remains the same.

}

  \begin {lemma}\label{lem:weights}
    \textlemmaweights
%
  \end {lemma}

\begin{proof}
  The value $\omega (\mathcal{R},\mathcal{S},d)$ is the maximum number
  of external bottom-level adjacencies between $\cal R$ and $\cal S$
  that can be realized if $\cal S$ is adjacent to $\cal R$ in
  direction $d$. By Lemma~\ref{lem:contiguos} there is a rectangular
  extension of the global layout $\cal L$ in which this number of
  adjacencies between $\cal R$ and $\cal S$ is realized. So clearly
  the difference in $\cal R$-$\cal S$ adjacencies is $\omega
  (\mathcal{R},\mathcal{S},d') - \omega
  (\mathcal{R},\mathcal{S},d)$. Adjacent pairs of rectangle complexes
  whose adjacency direction remains the same are not affected by
  changes of adjacency directions elsewhere in the global layout.
\end{proof}

This basically means we can consider changes of adjacency directions
locally and independent from the rest of the layout.
%
%
Furthermore, since the values $\omega (\mathcal{R},\mathcal{S},d)$ are
directly obtained from counting the numbers of $d$-extensible and $d$-engaged
rectangles in $\mathcal{R}$ (or $-d$-extensible and $-d$-engaged
rectangles in $\cal S$) we get the next lemma.

  \begin {lemma}
    We can compute all values $\omega (\mathcal{R},\mathcal{S},d)$ in
    $O (n)$ total time.
  \end {lemma}

\subsubsection {Optimizing in a graph without separating 4-cycles}

  Here we will prove the following:


  \begin {theorem} \label {thm:no4cycles}
    Let $G$ be an embedded triangulated planar graph with
    $k'$ vertices without separating $3$-cycles and without
    non-trivial separating $4$-cycles, except for the outer face which
    consists of exactly four vertices. Furthermore, let a weight
    $\omega (e, d)$
    be assigned to every edge $e$ in $G$ and every orientation $d$ in
    $\D$.  Then we can find a rectangular subdivision of which $G$ is
    the extended dual that maximizes the total weight of the directed
    adjacencies in $O (k'^4 \log k')$ time. 
  \end {theorem}


  In order to optimize over all rectangular subdivisions with the same
  extended dual graph we make use of the representation of these
  subdivisions as elements in a distributive lattice or, equivalently,
  as closures in a partial order induced by this lattice~\cite{emsv-aurl-09,em-ocrl-09}. There are two \emph{moves}
  or \emph{flips} by which we can transform one rectangular layout (or
  its regular edge labeling) into another one,  \emph{edge flips} and \emph{vertex flips} (Figure~\ref{fig:flips}). They form a graph where each
  equivalence class of rectangular layouts is a vertex and two
  vertices are connected by an edge if they are transformable into
  each other by a single move, with the edge directed toward the more
  counterclockwise layout with respect to this move. This graph is
  acyclic and its reachability ordering is a distributive lattice~\cite{Fus-DM-09}. It
  has a minimal (maximal) element that is obtained by repeatedly
  performing clockwise (counterclockwise) moves.

  By Birkhoff's representation theorem~\cite{Bir-DMJ-37} each element in
  this lattice is in one-to-one correspondence to a partition of a
  partial order $\cal P$ into an upward-closed set $U$ and a
  downward-closed set $L$. The elements in $\cal P$ are pairs $(x,i)$,
  where $x$ is a flippable item, i.e., either the edge of an edge flip
  or the vertex of a vertex flip~\cite{emsv-aurl-09,em-ocrl-09}. The integer $i$ is the so-called
  flipping number $f_x(\cal L)$ of $x$ in a particular layout $\cal
  L$, i.e., the well-defined number of times flip $x$ is performed
  counterclockwise on any path from the minimal element
  $\mathcal{L}_{\min}$ to $\cal L$ in
  the distributive lattice. An element $(x,i)$ is smaller than another
  element $(y,j)$ in this order if $y$ cannot be flipped for the
  $j$-th time before $x$ is flipped for the $i$-th time. For each
  upward- and downward-closed partition $U$ and $L$, the corresponding
  layout can be reconstructed by performing all flips in the lower set
  $L$. $\cal P$ has $O(k'^2)$ vertices
  and edges and can be constructed in $O(k'^2)$ time~\cite{emsv-aurl-09,em-ocrl-09}. The construction
  starts with an arbitrary layout, performs a sequence of clockwise
  moves until we reach $\mathcal{L}_{\min}$, and from there performs a
  sequence of counterclockwise moves until we reach the maximal
  element. During this last process we count how often each element is
  flipped, which determines all pairs $(x,i)$ of $\cal P$. Since each
  flip $(x,i)$ affects only those flippable items that belong to the
  same triangle as $x$, we can initialize a queue of possible flips,
  and iteratively extract the next flip and add the new flips to the
  queue in total time $O(k'^2)$. In order to create the edges in $\cal
  P$ we again use the fact that a flip $(x,i)$ depends only on flips
  $(x',i')$, where $x'$ belongs to the same triangle as $x$ and $i'$
  differs by at most 1 from $i$. The actual dependencies can be
  obtained from their states in $\mathcal{L}_{\min}$.

\begin{figure}[tb]
  \centering
  \subfigure[edge flip $BD$]{\includegraphics[page=1]{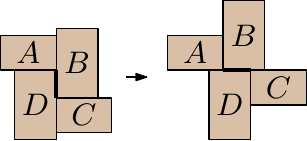}}
  \hfil
  \subfigure[vertex flip $E$]{\includegraphics[page=2]{figures/flips}}
  \caption{Flip operations}
  \label{fig:flips}
\end{figure}


Next, we assign weights to the nodes in $\cal P$. Let $\cal L_{\min}$
be the layout that is minimal in the distributive lattice, i.e., the
layout where no more clockwise flips are possible. For an edge-flip
node $(e,i)$ let $\cal R$ and $\cal S$ be the two rectangle complexes
adjacent across $e$. Then the weight $\omega(e,i)$ is obtained as
follows. Starting with the adjacency direction between $\cal R$ and
$\cal S$ in $\cal L_{\min}$ we cycle $i$ times through the set $\D$ in
counterclockwise fashion. Let $d$ be the $i$-th direction and $d'$ the
$(i+1)$-th direction. Then $\omega(e,i) = \omega(e,d') =
\omega(\mathcal{R},\mathcal{S},d') -
\omega(\mathcal{R},\mathcal{S},d)$. For a vertex-flip node $(v,i)$ let
$\cal R$ be the degree-4 rectangle complex surrounded by the four
complexes $\mathcal{S}_1, \dots, \mathcal{S}_4$. We again determine
the adjacency directions between $\cal R$ and $\mathcal{S}_1, \dots,
\mathcal{S}_4$ in $\cal L_{\min}$ and cycle $i$ times through $\D$ to
obtain the $i$-th directions $d_1, \dots, d_4$ as well as the $(i+1)$-th
directions $d'_1, \dots, d'_4$. Then $\omega(v,i) = \sum_{j=1}^4
\omega(\mathcal{R},\mathcal{S}_j,d'_j) -
\omega(\mathcal{R},\mathcal{S}_j,d_j)$. Equivalently, if the four edges incident to
$v$ are $e_1, \dots, e_4$, we have $\omega(v,i) = \sum_{j=1}^4
\omega(e_j,d'_j)$.



Finally, we compute a maximum-weight closure of $\cal P$ using a
max-flow algorithm~\cite[Chapter 19.2]{amo-nf-93}, which will take $O
(k'^4 \log k')$ time for a graph with $O(k'^2)$ nodes.


\subsubsection {Optimizing in General Graphs}

In this section, we show how to remove the restriction that the graph
should have no separating $4$-cycles. We do this by decomposing the graph $G$ by its separating $4$-cycles and solving the subproblems in a bottom-up fashion.

  \tweeplaatjes {dectree-graph} {dectree-tree-hor} {(a) A graph with non-trivial separating $4$-cycles. Note that some $4$-cycles intersect each other. (b) A possible decomposition tree of $4$-cycle-free graphs (root on the left).}

  \begin {lemma}[Eppstein et al.~\cite{emsv-aurl-09}]
    Given a plane graph $G$ with $k$ vertices, there exists a
    collection $\cal C$ of separating 4-cycles in $G$ that decomposes
    $G$ into separation components that do not contain separating
    4-cycles any more. Such a collection $\cal C$ and the
    decomposition can be computed in $O(k)$ time.
    %
  \end {lemma}

  These cycles naturally subdivide $G$ into a tree of subgraphs, which
  we will denote as $T_G$. Still following~\cite {emsv-aurl-09}, we
  add an extra artificial vertex inside each $4$-cycle, which
  corresponds to filling the void in the subdivision after removing
  all rectangles inside by a single rectangle.  Figure~\ref
  {fig:dectree-graph+dectree-tree-hor} shows an example of a graph $G$ and
  a corresponding tree $T_G$.

  Now, all nodes of $T_G$ have an associated graph without separating
  $4$-cycles on which we can apply Theorem~\ref {thm:no4cycles}. The
  only thing left to do is assign the correct weights to the edges of
  these graphs.  For a given node $\nu$ of $T_G$, let $G_\nu$ be the
  subgraph of $G$ associated to $\nu$ (with potentially extra vertices
  inside its $4$-cycles).

  For every leave $\nu$ of $T_G$, we assign weights to the internal
  edges of $G_\nu$ by simply setting $\omega(e,d) = \omega (\mathcal
  R, \mathcal S, d)$ if $e$ separates $\mathcal R$ and $\mathcal S$ in
  the global layout $\mathcal L$.  For the external edges of $G_\nu$
  (the edges that are incident to one of the ``corner'' vertices of
  the outer face), we fix the orientations in the four possible ways,
  leading to four different problems. We apply Theorem~\ref
  {thm:no4cycles} four times, once for each orientation. We store the
  resulting solution values as well as the corresponding optimal
  layouts at $\nu$ in $T_G$.

  Now, in bottom-up order, for each internal node $\nu$ in $T_G$, we
  proceed in a similar way with one important change: for each child
  $\mu$ of $\nu$, we first look up the four optimal layouts of $\mu$
  and incorporate them in the weights of the four edges incident to
  the single extra vertex that replaced $G_\mu$ in $G_\nu$. Since
  these four edges must necessarily have four different orientations,
  their states are linked, and it does not matter how we distribute
  the weight over them; we can simply set the weight of three of these
  edges to $0$ and the remaining one to the solution of the
  appropriately oriented subproblem. The weights of the remaining
  edges are derived from $\mathcal L$ as before, and again we fix the
  orientations of the external edges of $G_\nu$ in four different ways
  and apply Theorem~\ref {thm:no4cycles} to each of them. We again
  store the resulting four optimal values and the corresponding
  layouts at $\nu$, in which we insert the correctly oriented
  subsolutions for all children $\mu$ of $\nu$.


  This whole process takes $O (k^4 \log k)$ time in the worst case.
  Finally, since weights are expressed as differences with respect to
  the minimal layout $\mathcal{L}_{\min}$ we compute the value of
  $\mathcal{L}_{\min}$ and add the offset computed as the optimal
  solution to get the actual value of the globally optimal
  solution. This takes $O (n)$ time.

  \begin {theorem} \label {thm:with4cycles} Let $\mathcal T$ be a
    $2$-level treemap, such that the extended dual graph $G$ of the
    global layout has no separating $3$-cycles. Let $n$ be the number
    of cells in the bottom level and $k$ the number of regions in the
    top level.  Then we can find a rectangular subdivision that
    preserves all oriented internal bottom-level adjacencies, and
    preserves as many oriented external bottom-level adjacencies as
    possible in $O (k^4 \log k + n)$ time. 
  \end {theorem}

\section {Without preserving orientations}


  In this section we study the variant of the problem where we do not
  need to preserve the orientations of the adjacencies that we preserve.
  We still assume that 
  the required global layout of the output treemap is
  given in advance.

We first define an adjacency graph on the boundary cells of all
  rectangle complexes. There is a vertex in this graph for each cell
  that belongs to the boundary of a rectangle complex.
  Since the global
  layout is given, we know where the four corners separating the
  boundary sides are. There is an \emph{internal adjacency} edge
  between any two vertices of the same rectangle complex whose cells
  are adjacent in the complex.
  There are \emph{external adjacency} edges between vertices of
  cells of different rectangle complexes if the cells are adjacent in
  the input.

Next we note that the adjacency graph as a 
  subgraph of the dual graph for the input layout is a planar graph.
Our goal is to select subsets of the vertices to be on the
  boundary of each rectangle complex whose induced subgraph has as
  many external adjacency edges as possible but also would not create
  any separating triangles if we imagine connecting all the external adjacencies corresponding to one boundary to one vertex.

  For the remainder we restrict us to the case of two top-level regions, and only at the very end extend the arguments to more regions.
  In our subgraph of the adjacency graph we remove those internal adjacency edges that
  correspond to two directly neighboring cells when traversing the
  boundary of the corresponding rectangle complex.
    In the remaining
  graph, we need to find an independent set in terms of the remaining
  internal adjacency edges, since any two adjacent vertices in this
  graph would induce a separating triangle. We first prove that preserving as many bottom-level external adjacencies as possible is NP-hard already for the case of two top-level regions.

\subsection{NP-hardness}
In positive 1-in-3-SAT each clause contains exactly three (non-negated) variables, and we need to decide whether there is a truth assignment such that exactly one variable per clause is true. As input we are given the collection of clauses together with a planar embedding of the associated graph such that all variables are on a straight line and no edge crosses the straight line. This problem was shown to be NP-complete by Mulzer and Rote~\cite{mr-mwt-08}.
We reduce from the following variant of positive 1-in-3-SAT.
\begin{lemma}\label{lem:satvariant}
Planar positive 1-in-3-SAT with variables on a line and with every variable occurring in at least one clause on each side of the line and in at most three clauses is NP-hard.
\end{lemma}
\eenplaatje {equality-gadget} {equality gadget $x \Leftrightarrow z$}
\begin{proof}
Consider the \emph{equality gadget} in Figure~\ref{fig:equality-gadget}. It enforces that $x$ and $z$ are equivalent. The equality gadget consists of the same clauses as the corresponding gadget in~\cite{mr-mwt-08}, but it arranges them such that any variable other than $x$ and $z$ occurs in a clause on each side of the line. Concatenating equality gadgets gives us a sequence of equivalent variables $x_0, \ldots, x_{k+1}$, where $x_0$ and $x_{k+1}$ occur in exactly one clause and the other variables occur in exactly two clauses, one clause on each side of the line. Given an instance of planar positive 1-in-3-SAT, we replace any variable that occurs in $k>1$ clauses by such a sequence, and use $x_1,\ldots,x_k$ to connect to one additional clause each. We connect any remaining variable $x$ that only occurs in one clause (like $x_0$ and $x_{k+1})$ to two new identical clauses $(x \vee a \vee b)$ and $(x \vee a \vee b)$ with additional variables $a$ and $b$, one on each side of the line.
\end{proof}

  \begin{figure}[htb]
    \centering
    \includegraphics[scale=0.5]{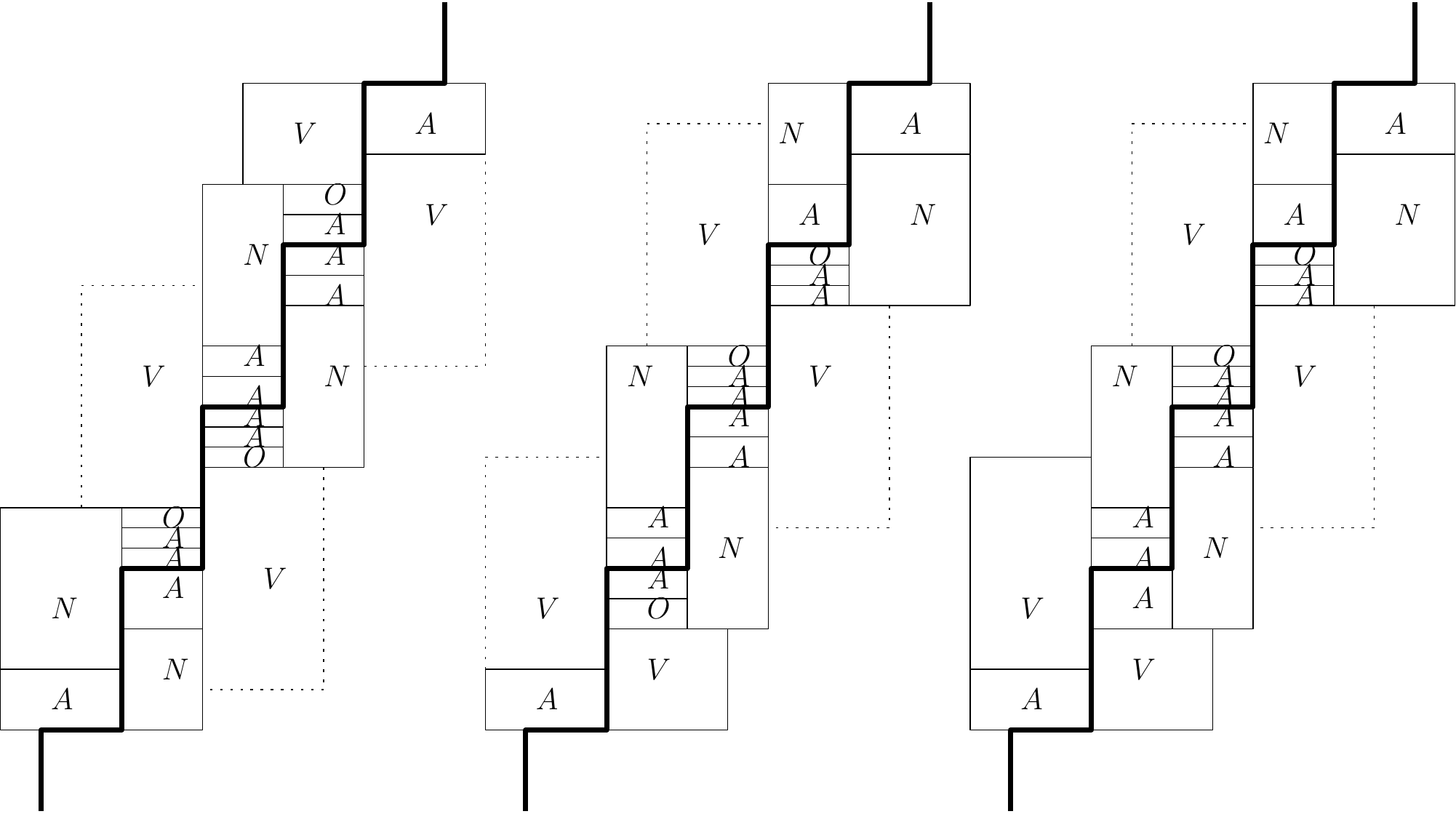}
    \caption{A set of rectangles representing a variable. Left: 1 occurrence above the line, 2 below; middle: 2 above, 1 below; right: 1 and 1.}
    \label{fig:variable-gadget}
  \end{figure}
In the reduction we have two top-level regions with a staircase boundary between them. For each variable we have a variable gadget consisting of a set of rectangles on both sides of the boundary. Figure~\ref{fig:variable-gadget} shows the three variable gadgets used. Which gadget is used, depends on the number of occurrences in clauses. From the $V$- and $N$-rectangles we can on each side of the boundary only keep an independent set (in terms of the adjacency graph) since any adjacent pair would induce a separating triangle due to the $A$-rectangles. A $V$-rectangle extends beyond its vertical and/or its horizontal dotted line; we call such a $V$-rectangle \emph{connecting}. $O$-rectangles are only placed opposite to those $V$-rectangles that are extended in this way. The complete boundary is filled with vertex gadgets.

A clause is represented by three adjacencies: for a clause on the left of the boundary, from the left variable gadget a $V$-rectangle extends upwards and from the right variable gadget a $V$-rectangle extends to the left, so that the two rectangles touch. From the middle vertex gadget a $V$-rectangle extends to the left and upwards, so that it touches the two other rectangles. The case to the right of the boundary is analogous. We do this for every clause. Finally any empty spaces are filled with rectangles greedily.

For the moment assume that for a variable either all $V$-rectangles (and no $N$-rectangles) are kept on the boundary (this corresponds to setting a variable to true) or all $N$-rectangles (and no $V$-rectangles) are kept (this corresponds to setting a variable to false). Additionally we keep as many $A$- and $O$- rectangles as possible. In the case that we keep all $V$-rectangles we can achieve 10 or 11 adjacencies depending on whether the variable occurs 2 or three times. Thus, we achieve 8 adjacencies plus the number of occurrences of the variable in clauses. In the case that we keep all $N$-rectangles we can achieve 8 adjacencies. Additionally, we get 1 adjacency for each pair of neighboring vertex gadgets. For the three variable gadgets involved in the same clause gadget, we can for at most one keep all $V$-variables since the clause gadget would otherwise yield a separating triangle. Thus if we have $m$ variables and $n$ clauses, the total number of adjacencies is $8m + m-1$ plus the number of occurrences of variables that have been set to true. Now if the original formula has a satisfying assignment with exactly one variable true per clause then the number of adjacencies we can achieve in this way is $9m -1 +n$, and if there is no such assignment the number of adjacencies is smaller.

It remains to show that we can indeed assume that for a vertex gadget we have only $V$-rectangles or only $N$-rectangles. First observe that as long as we keep all connecting $V$-rectangles on the boundary, there is no advantage in dropping any of the remaining ones. It remains to prove that we do not get more than 8 adjacencies if we do not choose to keep all connecting $V$-rectangles on the boundary; this implies that such a configuration has no advantage over the one with all $N$-rectangles, thus we can replace it by the later.

We refer to external adjacencies by the symbols of the rectangles, e.g., we call an adjacency between a $V$- and an $A$-rectangle a $V$-$A$ adjacency.
We now go through all cases. If we keep the $N$-rectangles on 1 side and the $V$-rectangles on the other side then we get no $V$-$V$ or $N$-$N$ adjacencies, 3 $V$-$A$, 3 $N$-$V$ and at most 2 $V$-$O$ adjacencies, thus at most 8. There is one configuration in which we keep a pair of externally-adjacent $V$-rectangles and a pair of externally-adjacent $N$-rectangles. In this case we have 1 $V$-$V$, 1 $N$-$N$, 2 $V$-$A$, 2 $N$-$A$, and at most 2 $V$-$O$ adjacencies, thus at most 7. If we keep $3$ $V$-rectangles, but not all connecting ones, then we loose at least 3 adjacencies, so we keep at most 8. All other cases give subsets of the cases above, and leave us with even fewer (at most 7) adjacencies. Thus, the assumption was valid. From this the following theorem follows.

    \begin {theorem}
      Given an input subdivision, 
      finding a rectangular subdivision that respects the global layout, preserves all internal adjacencies, and preserves as many bottom-level external adjacencies as possible is NP-hard.
    \end {theorem}

\subsection {Upper bound}

  We now show that it is sometimes not possible to preserve more than a factor $1/4$ of the external adjacencies.


  We will construct an example with $4h+7$ boundary rectangles (and therefore $4h+6$) external adjacencies between two regions. To be more precise, we will construct a graph which we show is the dual graph of a pair of rectangle complexes.
  Figure~\ref {fig:upper-bound-4} illustrates the construction.
  The graph consists of two pieces (one for each rectangle complex). The top piece has $7$ vertices, $3$ of which are isolated and $4$ of which are connected by a maximal outerplanar graph.
  The bottom piece consists of $4h$ vertices, in four groups of $h$. The first and third groups of vertices are all isolated, while the vertices of the second and fourth groups are pairwise connected. Finally, we add a complete planar bipartite graph between the two groups of vertices by connecting each group of $h$ vertices to one of the $4$ connected vertices in the top graph, and filling the gaps with $6$ additional edges.

  \begin {claim}
    We cannot preserve more than $h+6$ edges in this construction.
  \end {claim}

  \begin {proof}
    Recall that we need to take an independent set in both the top and bottom outerplanar graph.
    This implies that in the top graph, if we keep the first or the fourth vertex, we cannot keep any of the other three. Therefore, we can either keep the complete first or the third group of $h$ edges, or a combination of edges in the second and fourth group together. However, each edge in the second group is connected to an edge in the fourth group via an internal adjacency in the bottom graph, so we can keep at most half of the edges in these groups together: at most $h$.
    Since there are only $6$ edges not in a group, we can keep at most $h+6$ edges.
  \end {proof}

  If desired, we can create another copy of the construction and place it upside down to have a non-constant number of vertices at the top as well as at the bottom.

\eenplaatje [scale=0.8] {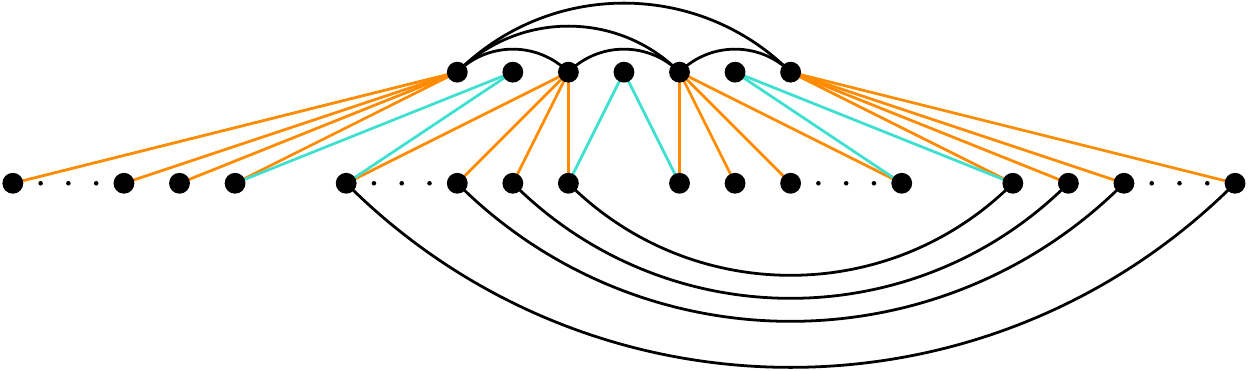} {Illustration of the construction showing that we cannot keep more than a quarter of the external adjacencies.
Black edges indicate internal adjacencies; orange and blue edges are external adjacencies. Blue edges are incident to an isolated vertex in the top graph.}

    \begin {theorem}
      Given an input subdivision, 
      there generally exists no rectangular subdivision that respects the global layout, preserves all internal adjacencies, and preserves more than $1/4$th of the bottom-level external adjacencies.
    \end {theorem}

\subsection{Algorithm}
We first describe an algorithm for two top-level regions. The connected components
  of internal adjacency edges form outerplanar graphs on which we can
  solve the maximum weight independent set problem in linear time exactly (where
  the weight of a vertex is its degree in terms of external adjacency
  edges). We first solve the maximum weight independent set problem for one of the sides of the boundary exactly and then solve it for the other side using only the adjacencies with rectangles from the maximum weight independent set.

   By three-coloring the outerplanar graph we see that the weight of the independent set is at least a third of the total weight since the maximum weight of a colour class is a lower bound for the weight of the independent set. Applying the same argument to the other outerplanar graph, shows that we preserve at least a third of the remaining adjacencies. Thus overall we keep at least $1/9$th of the external adjacencies. Furthermore, the weight of the independent set for the first side is an upper bound on the number of external adjacencies that we can preserve. Since we preserve a third of this weight, our algorithm is a $1/3$-approximation.

For more than two top-level regions the choices of which rectangles to place on a boundary are not independent. Instead of solving the problem for the boundaries between two regions we solve it for the line segments of the global layout (with possibly more than one region on each of the sides of the line segment). The choices of which rectangles to keep adjacent to line segments are also not independent, but if we only consider horizontal or only vertical line segments they are. By optimizing only for horizontal or only for vertical line segments, we again loose at most a factor of $1/2$ in terms of the number of adjacencies preserved. Thus, overall we can preserve at least $1/18$th of the external adjacencies and obtain a $1/6$-approximation. We can compute a corresponding rectangular subdivision in linear time~\cite{kh-tafrd-93}.

   \begin {theorem}
      Given an input subdivision, 
      we can find a rectangular subdivision that respects the global layout, preserves all internal adjacencies, and preserves at least $1/18$th of the bottom-level external adjacencies in $O (n)$ time.
    \end {theorem}

    \begin {corollary}
      Given an input subdivision, 
      let $s$ be the maximal number of internal adjacencies preserved in any rectangular subdivision that respects the global layout, preserves all internal adjacencies.
      We can find a rectangular subdivision that respects the global layout, preserves all internal adjacencies, and preserves at least $1/6 s$ bottom-level external adjacencies in $O (n)$ time.
    \end {corollary}

\section* {Acknowledgements}

This research was initiated at MARC 2009. We would like to thank Jo Wood for proposing this problem, and all participants for sharing their thoughts on this subject.

D.~E. is supported by the National Science
Foundation under grant 0830403. D.~E. and M.~L. are supported by the
U.S. Office of Naval Research under grant N00014-08-1-1015. M.~N. is
supported by the German Research Foundation (DFG) under grant
NO 899/1-1. R.~I.~S. is supported by the Netherlands Organisation for Scientific Research (NWO).

  \bibliographystyle{abbrv}
  \bibliography{refs}

\end{document}